\documentclass[a4paper]{article}
\usepackage[left=2cm,right=2cm,top=2.5cm,bottom=2.5cm]{geometry}
\usepackage{amssymb}
\usepackage{physics}
\usepackage{amsmath}
\usepackage{amsthm}
\usepackage{graphicx}
\usepackage{mathrsfs}
\usepackage[colorlinks=true,linkcolor=blue,citecolor=blue,urlcolor=blue]{hyperref}

\renewcommand{\d}{\mathrm{d}}
\newcommand{\e}{\mathrm{e}}
\newtheorem{thm}{Theorem}

\title{Quantum speed limits based on quantifiers of quantum-state texture}
\author{Yuhang Xie$^1$, Yanjun Chu$^1$\thanks{Corresponding author. Email: chuyj@henu.edu.cn}, Chenyang Cui$^1$, Shao-Ming Fei$^2$ \\[8pt]
$^1$School of Mathematics and Statistics, Henan University,\\
Kaifeng, 475004, China \\
$^2$School of Mathematical Sciences, Capital Normal University, \\
Beijing, 100048, China}
\date{}

\begin{document}

\maketitle

\begin{abstract}
Quantum speed limits impose intrinsic lower bounds on the shortest time scale for quantum system evolution. As an emerging paradigm in quantum resource theory, quantum-state texture has attracted  research interest amid the rapid advancement of quantum theory. Herein, we investigate the interplay between quantum speed limits and quantum-state texture via several canonical quantifiers, including trace distance, state rugosity and Jensen-Shannon divergence. To demonstrate our findings, we analyze the minimum evolution time of physical systems subject to dephasing and dissipative dynamics. For the Jensen-Shannon divergence, we further explore nonunitary dynamics described by completely positive and trace-preserving maps, taking the amplitude damping channel as a typical example. In addition, we explore the tightness of these bounds in the considered dynamical models. Our results reveal that quantum speed limits derived from quantum-state texture capture the fundamental constraints on quantum evolutionary speed, with promising applications in quantum computing, quantum control and quantum metrology.
\end{abstract}
\section{Introduction}\label{sect1}
The quantum speed limit (QSL) provides a fundamental bound on the evolution time of quantum systems undergoing arbitrary physical processes \cite{MT,ML,anandan1990,levitin2009,shanahan2018}, and plays a crucial role in advanced quantum technologies, encompassing optimal control  \cite{caneva2009,kobayashi2020}, quantum metrology \cite{giovannetti2006,toth2014JPA,macri2016PRA,campbellQST025002}, quantum thermodynamics \cite{deffner2010PRL,campisi2011RMP,goold2016JPA,aghion2023JPA,hasegawa2023NC}, communication and quantum computing \cite{bekenstein2000Nature,lioyd2000Nature}, as well as quantum batteries \cite{mohan2021PRA,gyhm2024PRA,juliafarre2020PRR}. For closed quantum systems, Mandelstam and Tamm \cite{MT} established a fundamental lower bound on the unitary evolution time connecting the initial state $\psi_0$ to the final state $\psi_{\tau}$. Specifically, the evolution duration satisfies $\tau\geqslant\tau_{MT}:=(\hbar/\Delta E)\arccos(|\braket{\psi_0}{\psi_{\tau}}|)$, where $\Delta E=\sqrt{\langle H^2\rangle-\langle H\rangle^2}$ denotes the variance of a time-independent Hamiltonian $H$ governing the unitary dynamics. In the special case of orthogonal states transition with $\braket{\psi_0}{\psi_{\tau}}=0$, this bound reduces to $\tau_{MT}=\pi\hbar/2\Delta E$. Subsequently,  Margolus and Levitin \cite{ML} developed an alternative QSL bound, $\tau\geqslant\tau_{ML}=\pi\hbar/2(\langle\psi_0|H|\psi_0\rangle-E_0)$, with $E_0$ corresponding to the ground-state energy of the quantum system.

In Ref.~\cite{rudnicki2021PRA}, the QSL is investigated in terms of geometric measure of entanglement. Refs.~\cite{mohan2022NJP,mai2024PRA} further derived QSLs based on quantum coherence, uncovering deep connections between quantum coherence and QSLs. As coherence depends on the choice of basis,  Xiao \textit{et al.} \cite{xiaoPRA032438} studied the QSL governed by coherence measures defined in the basis associated with the spectral decomposition of the time-dependent Hamiltonian. In this scenario, the incoherent states become time-dependent and satisfy  \([H_t, \sigma_t] = 0\) for   the
time-dependent Hamiltonian $H_t$ and  arbitray incoherent state $\sigma_t$. Correspondingly, their coherence measures are also time-dependent quantities that characterize the distance between the evolved state \(\rho_t\) and the instantaneously evolving incoherent set. Moreover, Xuan \textit{et al.}~\cite{xuan2025PRA} explored the evolution times of quantum systems by deriving QSL bounds based on three distinct measures of imaginarity. The generation and degradation of imaginarity are studied under several typical quantum dynamical scenarios including dephasing dynamics, dissipative dynamics and stochastic-approximate transformations. Correspondingly, the concept of the resource speed limit (RSL) has recently been introduced in Ref. \cite{Campaioli2022}, which quantifies the maximum rate at which quantum resources can be generated or degraded by physical  processes. Meanwhile, we also note that various entropic quantities have recently been adopted to establish QSLs, including the unified entropy \cite{piresPRA012403}, the generalized relative entropy \cite{sousaPRA012203}, and the Sharma–Mittal entropy \cite{xuane00383}. Quantum speed limits based on the Jensen–Shannon divergence and Jeffreys divergence \cite{sousaPRA042419}, as well as those formulated in terms of various fidelity measures \cite{osanPRA022443}, have also been explored.

The notion of quantum-state texture (QST) was first proposed by F. Parisio \cite{parisio2024} as a novel quantum resource intrinsically associated with coherence and imaginarity. A straightforward interpretation of QST is as follows: For a given computational basis $\{\ket{i}\}$, the density matrix $\rho$ can be visualized as a three-dimensional plot, where the row and column indices of the matrix correspond to the horizontal dimensions, and the magnitudes of the real or imaginary parts of each matrix entry $\rho_{ij}$ serve as the vertical altitude in two complementary subplots (see FIG. 1 of Ref. \cite{parisio2024}). Within this geometric description, every quantum state corresponds to a pair of three-dimensional plots for the real and imaginary components of $\rho$. The inhomogeneity of these plots, i.e., their departure from a uniform distribution of matrix entries, characterizes the QST associated with the quantum state. Recently, the quantification and physical applications of texture resource theory have gained considerable attention \cite{wang2025,zhang2025,muthuganesan2025,huangQIP388,cao2026JPAMT,PatraPRA022411,greenwood2602.22496v1}.

We explore the connections between QSLs and QST via several representative QST measures, including the trace distance, state rugosity and Jensen–Shannon divergence. As practical applications, we explore the dephasing and dissipative dynamics using QSLs constructed from the trace distance and state rugosity, and analyze the tightness of QSLs under these dynamics processes. For the QST measure based on Jensen-Shannon divergence, we investigate the minimal time required for the evolution of physical systems subject to completely positive and trace-preserving (CPTP) maps, and take the amplitude damping channel as a concrete example to illustrate our findings. These QSLs derived from QST measures characterize the fundamental limits on the speed of quantum evolution, which reveals that QST plays a role in dynamical evolution analogous to quantum coherence \cite{mohan2022NJP,mai2024PRA} and iquantum maginarity \cite{xuan2025PRA}, with potential applications in quantum information processing.

This paper is organized as follows. In Sect. \ref{sect2}, we briefly review the resource theory of texture and present some QST measures adopted throughout this paper. In Sect. \ref{sect3}, we derive QSLs from three representative QST measures: trace distance, roughness and Jensen–Shannon divergence. We illustrate our results under dephasing dynamics and dissipative dynamics, as well as the amplitude damping channel. Finally, we conclude with a summary in Sect. \ref{sect4}.

\section{Preliminaries}\label{sect2}

Analogous to the coherence and imaginarity, the texture depends on the choice of the orthonormal reference basis $\{|i\rangle\}_{i=0}^{d-1}$ of the $d$-dimensional Hilbert space $\mathscr{H}$. Denote $\mathcal{D}(\mathscr{H})$ the collection of all quantum states. In the resource theory of quantum-state texture, the only free state \cite{ChitambarGour2019} is the textureless state \cite{parisio2024},
\begin{equation}\label{textureless state}
f=\ketbra{f}{f},
\end{equation}
where $|f\rangle = \frac{1}{\sqrt{d}} \sum_{i=0}^{d-1} |i\rangle$. An operation $\Lambda$ given by the Kraus operators $\{K_n\}$ is called free if it preserves the textureless state $f$, i.e.,
$\Lambda(f)=\sum_n K_n f K_n^\dagger=f$, where $\sum_n K_n^\dagger K_n=\mathbb{I}$ with $\mathbb{I}$ denoting the identity operator. Since $f$ is a pure state, it is an extreme point of the compact convex set $\mathscr{D}(\mathcal{H})$. Thus, a quantum operation  $\Lambda$ is free if and only if $K_n f k_n^\dagger=p_n f$ for all $n$, namely,
$K_n |f\rangle=\alpha_n |f\rangle$, where $p_n=\Tr(K_n f K_n^\dagger)$ and $\sum_n |\alpha_n|^2=1$  for $\alpha_n\in\mathbb{C}$.

Based on the definition of well-defined measures of QST \cite{parisio2024}, several QST measures have been proposed and investigated in Refs.~\cite{parisio2024,wang2025,zhang2025,muthuganesan2025}:

\textit{Trace-distance measure} \cite{wang2025},
\begin{equation}
\mathcal{T}_{\tr}(\rho)=\frac{1}{2}\norm{\rho-f}_1,
\end{equation}
where $\norm{X}_1=\Tr\sqrt{XX^\dagger}$ is the trace norm of matrix $X$.

\textit{State rugosity} \cite{parisio2024},
\begin{equation}
\mathfrak{R}(\rho)=-\ln\bra{f}\rho\ket{f}.
\end{equation}

\textit{Measure based on Jensen–Shannon divergence}~\cite{muthuganesan2025},
\begin{equation}\label{JSD}
\mathcal{T}_{JS}(\rho)=JS(f, \rho),
\end{equation}
where  $JS(\sigma,\rho)=\frac{1}{2}\left[S\left(\sigma\middle\|\frac{\sigma+\rho}{2}\right)+S\left(\rho\middle\|\frac{\sigma+\rho}{2}\right)\right]=S\left(\frac{\sigma+\rho}{2}\right)-\frac{1}{2}S(\sigma)-\frac{1}{2}S(\rho)$ is the Jensen–Shannon divergence between the states $\sigma$ and $\rho$.

In the texture resource theory, the maximal texture states $\{f_k\}$ are referred to as Fourier states \cite{parisio2024},
\begin{equation}\label{Fourier states}
\ket{f_k}=\frac{1}{\sqrt{d}}\sum_{j=1}^{d}\omega^{(k-1)(j-1)}\ket{j},
\end{equation}
where $\omega=\mathrm{e}^{\mathrm{i}\frac{2\pi}{d}}$, $k>1$. Any arbitrary state can be obtained from $\{f_k\}$ via a free operations.

\section{Quantum speed limits based on quantum-state texture}\label{sect3}

In this section, we systematically investigate the intrinsic relationships between the QST and QSLs. We will see that the minimal time required for the evolution of a quantum system is bounded by the absolute difference in QST measures between the initial and final states, which demonstrates that the texture of quantum states plays a non-negligible role in the dynamical evolution of physical systems.

The time evolution of a quantum system under a specified dynamical process is governed by the master equation \cite{sudarshan1961},
\begin{equation*}
\dot{\rho}_t:=\frac{\mathrm{d}\rho_t}{\mathrm{d}t}=\mathcal{L}_t(\rho_t),
\end{equation*}
where $\rho_t$ represents the instantaneous state at time $t$, while $\mathcal{L}_t$ refers to the Liouvillian superoperator \cite{rivas}. For closed quantum systems governed by a time‑independent Hamiltonian $H$, the state evolves unitarily as $\rho_t=\e^{-\mathrm{i}Ht}\rho_0\e^{\mathrm{i}Ht}$.

\subsection{QSL from the trace-distance of QST}
The trace distance of QST serves as an important quantity for characterizing the texture structure of a quantum state \cite{wang2025}. Here we investigate the influence of the trace distance of texture on the minimal evolution time during the dynamical evolution of physical systems, i.e., its constraint on the evolution speed of physical systems.

\begin{thm}\label{QSL trace-distance}
Given an arbitrary quantum dynamics governing the time evolution of the state of a finite-dimensional quantum system. Denote $\tau$ the time required for the system to evolve from the initial state $\rho_0$ at $t=0$ to the final state $\rho_{\tau}$ at $t=\tau$. The minimal time is lower bounded by
\begin{equation}\label{bound QSL trace-distance}
\tau\geqslant \tau_{\tr}^{\mathrm{QSL}}=\frac{\abs{\Delta_{\tr}}}{\Lambda_{\tr}(\tau)},
\end{equation}
where $\Delta_{\tr}=\mathcal{T}_{\tr}(\rho_{\tau})-\mathcal{T}_{\tr}(\rho_0), \Lambda_{\tr}(\tau)=\frac{1}{\tau}\int_0^\tau\norm{\mathcal{L}(\rho_t)}_1\,\d t$.
\end{thm}

\begin{proof}
Using the fact that $\frac{\d}{\d x}\norm{A(x)}\leqslant\norm{\frac{\d}{\d x}A(x)}$ for any matrix norm $\norm{\cdot}$ and  any matrix function $A(x)$, we have
\begin{equation*}
\frac{1}{\tau}\int_0^\tau\norm{\dot{\rho}_t}_1\d t
\geqslant\frac{1}{\tau}\int_0^\tau\abs{\frac{\d}{\d t}\mathcal{T}_{\tr}(\rho_t)}\,\d t
\geqslant\frac{1}{\tau}\abs{\int_0^\tau\frac{\d}{\d t}\mathcal{T}_{\tr}(\rho_t)\,\d t}
=\frac{1}{\tau}\Big|\mathcal{T}_{\tr}(\rho_\tau)-\mathcal{T}_{\tr}(\rho_0)\Big|,
\end{equation*}
which gives rise to the bound \eqref{bound QSL trace-distance}.
\end{proof}

There generally exists a discrepancy between the derived lower bound of the evolution time in Eq. \eqref{bound QSL trace-distance} and the actual evolution time of a physical system. To quantify the deviation of the true physical value from the theoretical prediction, we mathematically define a quantity termed the normalized relative error to characterize such error behavior in the following:
\begin{equation}\label{tightness_tr}
\tilde{\delta}_{\tr}(\tau) := \frac{\delta_{\tr}(\tau)-\min(\delta_{\tr})}{\max(\delta_{\tr})-\min(\delta_{\tr})},
\end{equation}
where $\delta_{\tr}(\tau)$ is the unnormalized relative error,
\begin{equation*}
\delta_{\tr}(\tau) := 1 - \frac{\tau_{\tr}^{\text{QSL}}}{\tau}.
\end{equation*}

Obvioulsy, The normalized relative error satisfies $0 \leqslant \tilde{\delta}_{tr}(\tau) \leqslant 1$. Overall, a smaller normalized relative error defined in Eq. \eqref{tightness_tr} corresponds to a tighter lower bound on the QST  time given in Eq. \eqref{bound QSL trace-distance}.

In order to elaborate the physical implication and theoretical applicability of Theorem \ref{QSL trace-distance}, we take dephasing dynamical model into consideration for concrete illustration.

\textit{Dephasing dynamics}. We investigate a two‑level atom coupled to a bosonic reservoir, whose total Hamiltonian is given by
\begin{equation}\label{Hamiltonian dephasing}
    H_{\text{tot}} = \frac{1}{2} \omega \sigma_z + \sum_j \omega_j b_j^\dagger b_j + \sum_j g_j \sigma_z b_j^\dagger + \text{H.c.},
\end{equation}
where $\omega$ represents the atomic transition frequency, $\omega_j$ is the frequency of the $j$-th mode's harmonic oscillator, $g_j$ denotes the corresponding atom-mode coupling strength, and $b_j$ is the annihilation operator of the $j$-th mode in the bosonic reservoir. In Schr\"{o}dinger representation, one has the master equation governing the evolution of the atomic system \cite{Breuer2007},
\begin{equation}\label{master equation dephasing}
\dot{\rho}_t = -\mathrm{i}[H, \rho_t] + \frac{\gamma_t}{2} \left( \sigma_z \rho_t \sigma_z - \rho_t \right),
\end{equation}
where $H=\frac{1}{2}\omega\sigma_z$ represents the free Hamiltonian of the atom, and $\gamma_t$ denotes the dephasing rate. The time-dependent dephasing rate $\gamma_t$ in Eq. \eqref{master equation dephasing}
accounts for the impacts of the reservoir’s spectral properties as well as the coupling constants $g_j$. Deriving the effective master equation from the total Hamiltonian \eqref{Hamiltonian dephasing} follows a standard procedure within open quantum systems \cite{Breuer2007}. In the Markovian limit, the time‑dependent decay rate $\gamma_t$, simplifies to a constant $\gamma$, whereas it exhibits an explicit temporal dependence in the non‑Markovian regime \cite{Breuer2007,Das2021}.

Let the initial state be $\rho_{0}=\ketbra{\psi(0)}$, where the initial pure state reads $\ket{\psi(0)}=\cos\theta\ket{0}+\sin\theta\ket{1}$. The corresponding density matrix of the initial state can be expressed as
\begin{equation}\label{initial state}
\rho_0=
\begin{pmatrix}
\cos^2\theta & \cos\theta\sin\theta \\
\cos\theta\sin\theta & \sin^2\theta
\end{pmatrix}.
\end{equation}

Solving Eq. \eqref{master equation dephasing} we obtain
\begin{equation}\label{state_t dephasing}
\rho_t=
\begin{pmatrix}
\cos^2\theta & \cos\theta\sin\theta\,\e^{-\int_0^t \d s\,\gamma_s-\mathrm{i}\omega t} \\
\cos\theta\sin\theta\,\e^{-\int_0^t \d s\,\gamma_s+\mathrm{i}\omega t} & \sin^2\theta
\end{pmatrix}.
\end{equation}
We have
\begin{gather*}
\mathcal{T}_{\tr}(\rho_0)=\sqrt{2(1-\sin(2\theta))}, \\
\begin{aligned}
\mathcal{T}_{\tr}(\rho_t)=&\e^{-\int_0^t\d s\,\gamma_s-\frac{1}{2}\mathrm{i}\omega t} \\&\times\sqrt{2\e^{2\int_0^t\d s\,\gamma_s+\mathrm{i}\omega t}-\sin(2\theta)\e^{\int_0^t\d s\,\gamma_s+2\mathrm{i}\omega t}-\sin(2\theta)\e^{\int_0^t\d s\,\gamma_s}+\sin^2(2\theta)\e^{\mathrm{i}\omega t}-\sin^2(2\theta)\e^{2\int_0^t\d s\,\gamma_s+\mathrm{i}\omega t}},
\end{aligned}\\
\norm{\mathcal{L}_t(\rho_t)}_1=\abs{\sin(2\theta)}\sqrt{\gamma_t^2+\omega^2}\e^{-\int_0^t\d s\,\gamma_s},
\end{gather*}

In particular, we choose $\omega=0$ and $\gamma_t=2$, which is time independent and allows us to focus on the effect of the initial state's texture. The bound in Eq. \eqref{bound QSL trace-distance} becomes
\begin{equation}\label{bound_tr dephasing}
\tau_{\tr}^{\text{QSL}}=\frac{\abs{\e^{-2\tau}\sqrt{2\e^{4\tau}-2\sin(2\theta)\e^{2\tau}+\sin^2(2\theta)-\sin^2(2\theta)\e^{4\tau}}-\sqrt{2(1-\sin(2\theta))}}}{\frac{1}{\tau}\abs{\sin(2\theta)}(1-\e^{-2\tau})}.
\end{equation}

By the symmetry of the sine function in Eq.~\eqref{bound_tr dephasing}, we only need to analyze the cases over the two intervals $(0,\frac{\pi}{2})$ and $(\frac{\pi}{2},\pi)$. The quantum speed limits and normalized errors for initial states located in $(0,\frac{\pi}{2})$ and $(\frac{\pi}{2},\pi)$ are plotted in FIG.~\ref{dephasing_tr}, respectively.
\begin{figure}
\centering
\includegraphics[width=0.5\textwidth]{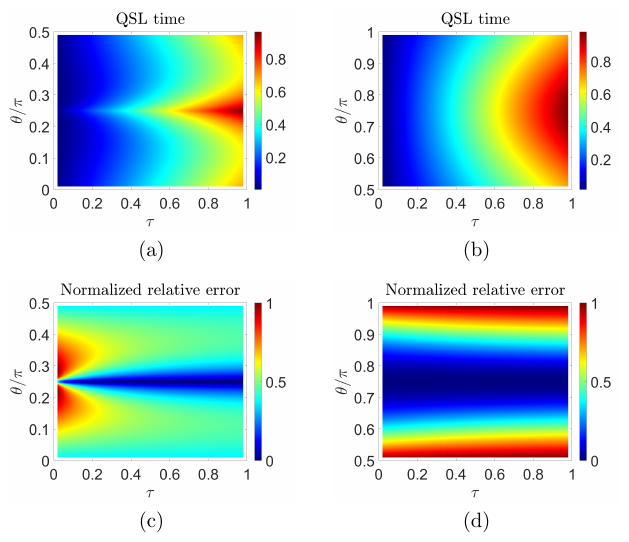}
\caption{The QSL time $\tau_{\tr}^{\text{QSL}}$ and normalized relative error $\tilde{\delta}_{\tr}(\tau)$ of dephasing dynamics with respect to different initial states.}
\label{dephasing_tr}
\end{figure}

We note that, for a fixed parameter $\theta$, $\tau_{\tr}^{\text{QSL}}$ monotonically increases as a function of the parameter $\gamma > 0$. FIG.~\ref{dephasing_tr} shows that the quantum speed limits and the normalized relative errors exhibit symmetry over the intervals $(0,\frac{\pi}{2})$ and $(\frac{\pi}{2},\pi)$, owing to that $\sin(2\left(\frac{\pi}{4}-\alpha\right))=\sin(2\left(\frac{\pi}{4}+\alpha\right))$ and $\sin(2\left(\frac{3\pi}{4}-\alpha\right))=\sin(2\left(\frac{3\pi}{4}+\alpha\right))$. In particular, FIG.~\ref{dephasing_tr}(a) presents that in the interval $(0,\frac{\pi}{2})$, the lower bound of the QSL gets larger when the textureless state $f$ in Eq.~(\ref{textureless state}) is adopted as the initial state, compared with other initial states; FIG.~\ref{dephasing_tr}(b) reveals that in the interval $(\frac{\pi}{2},\pi)$, when the maximal texture state
\begin{equation}\label{MTS}
\ket{g}=\frac{1}{\sqrt{2}}(\ket{0})-\ket{1})
\end{equation}
is chosen as the initial state, the lower bound of the QSL is larger than other initial states. FIG.~\ref{dephasing_tr}(c) shows that the normalized relative error $\tilde{\delta}_{\tr}(\tau)\approx 0$ when the parameter $\theta$ characterizing the initial state lies in the vicinity of $\frac{\pi}{2}$, which corresponds to the textureless state $f$ in Eq.~(\ref{textureless state}). This indicates that the lower bound of the QSL time exhibits a relatively high tightness. Similarly, FIG.~\ref{dephasing_tr}(d) show that the normalized relative error $\tilde{\delta}_{\tr}(\tau)\approx 0$ when the parameter $\theta$ lies in the vicinity of $\frac{3\pi}{4}$ which corresponds to the maximal texture state $\ket{g}$ in Eq.~(\ref{MTS}), i.e., $\frac{3\pi}{5}\lesssim\theta\lesssim\frac{9\pi}{10}$. This indicates that the lower bound of the QSL exhibits a relatively high tightness.

\subsection{QSL based on the state rugosity}

State rugosity is the first QST measure proposed by F. Parisio \cite{parisio2024}. He further adopted this quantity to demonstrate that QST can be utilized to characterize unknown quantum gates within universal circuit layers. Here, we employ state rugosity to characterize the constraint of QST on the minimal time required for the evolution of physical systems.

\begin{thm}\label{QSL rugosity}
For the dynamical evolution of a quantum system from the initial state $\rho_0$ to the final state $\rho_\tau$, the minimal time $\tau$ is bounded below by
\begin{equation}\label{bound QSL rugosity}
\tau\geqslant \tau_{\mathrm{ru}}^{\mathrm{QSL}}=\frac{\abs{\Delta_{\mathfrak{R}}}}{\Lambda_{\mathfrak{R}}(\tau)},
\end{equation}
where $\Delta_{\mathfrak{R}}=\mathfrak{R}(\rho_\tau)-\mathfrak{R}(\rho_0)$, $\Lambda_{\mathfrak{R}}(\tau)=\frac{1}{\tau}\int_0^\tau\frac{\bra{f}\,\abs{\mathcal{L}(\rho_t)}\,\ket{f}}{\bra{f}\rho_t\ket{f}}\,\d t$ and $\abs{\mathcal{L}(\rho_t)}$ denotes taking the absolute value of each entry in the matrix $\mathcal{L}(\rho_t)$, i.e., $(\abs{\mathcal{L}(\rho_t)})_{i,j}=\abs{(\mathcal{L}(\rho_t))_{i,j}}$.
\end{thm}

\begin{proof}
Since $\mathfrak{R}(\rho_t)=-\ln\bra{f}\rho_t\ket{f}$, one has
\begin{equation*}
\abs{\frac{\d}{\d t}\mathfrak{R}(\rho_t)}=\abs{-\frac{\frac{\d}{\d t}\bra{f}\rho_t\ket{f}}{\bra{f}\rho_t\ket{f}}}=\frac{\abs{\bra{f}\dot{\rho}_t\ket{f}}}{\bra{f}\rho_t\ket{f}}\leqslant\frac{\bra{f}\abs{\dot{\rho}_t}\ket{f}}{\bra{f}\rho_t\ket{f}}.
\end{equation*}
Therefore, by integrating over time $t$ from 0 to $\tau$, we have
\begin{equation*}
\frac{1}{\tau}\int_0^\tau\frac{\bra{f}\abs{\dot{\rho}_t}\ket{f}}{\bra{f}\rho_t\ket{f}}\,\d t
\geqslant \frac{1}{\tau}\int_0^\tau\abs{\frac{\d}{\d t}\mathfrak{R}(\rho_t)}\,\d t
\geqslant\frac{1}{\tau}\abs{\int_0^\tau\frac{\d}{\d t}\mathfrak{R}(\rho_t)\,\d t}
=\frac{1}{\tau}\Big|\mathfrak{R}(\rho_{\tau})-\mathfrak{R}(\rho_0)\Big|.
\end{equation*}
Hence, the conclusion holds.
\end{proof}

Similar to Eq.~\eqref{tightness_tr}, we define the normalized relative error with respect to the rugosity-based QSL time to characterize the tightness of the bound presented in Eq.~\eqref{bound QSL rugosity} as follows,
\begin{equation}\label{tightness_ru}
\tilde{\delta}_{\text{ru}}(\tau) := \frac{\delta_{\text{ru}}(\tau)-\min(\delta_{\text{ru}})}{\max(\delta_{\text{ru}})-\min(\delta_{\text{ru}})},
\end{equation}
where $\delta_{\text{ru}}(\tau)$ is the unnormalized relative error,
\begin{equation*}
\delta_{\text{ru}}(\tau) := 1 - \frac{\tau_{\text{ru}}^{\text{QSL}}}{\tau}.
\end{equation*}

Let us consider the following example to illustrate the applications of Theorem~\ref{QSL rugosity}.

\textit{Dissipative dynamics}. Suppose a two-level atom interacts with a leaky single-mode cavity \cite{Breuer2007,GarrawayPRA2290,DeffnerPRL010402}. The Hamiltonian of this system reads
\begin{equation}\label{Hamiltonian dissipative}
H_{\text{tot}}=\frac{1}{2}\omega\sigma_z+\sum_j\omega_jb_j^\dagger b_j+\sum_j g_j\sigma_+b_j+\text{H.c.},
\end{equation}
where these notations are defined analogous to those in Eq.~(\ref{Hamiltonian dephasing}). We assume the cavity is initially prepared in the vacuum state and functions as a memoryless Markovian reservoir. Under the Born-Markov approximation and rotating‑wave approximation, we trace out the cavity mode degrees of freedom to obtain the master equation \cite{LindabladCMP119,Lidar00967}
\begin{equation}\label{master equation dissipative}
\dot\rho_t=\gamma_t\left(\sigma_-\rho_t\sigma_+-\frac{1}{2}
\left\{\sigma_+\sigma_-,\rho_t\right\}\right),
\end{equation}
where $\gamma_t$ denotes the time-dependent decay rate. The operators $\sigma_{+}=\ketbra{0}{1}$ and $\sigma_{-}=\ketbra{1}{0}$ are the atomic raising and lowering operators, respectively.

We still consider the initial state of the system to be $\rho_0$ given by Eq. \eqref{initial state}. Solving Eq. \eqref{master equation dissipative} one derives
\begin{equation}\label{state_t dissipative}
\rho_t=
\begin{pmatrix}
\cos^2\theta\,\e^{-\int_0^t \d s\,\gamma_s} & \cos\theta\sin\theta\,\e^{-\int_0^t \d s\,\frac{\gamma_s}{2}} \\
\cos\theta\sin\theta\,\e^{-\int_0^t \d s\,\frac{\gamma_s}{2}} & 1-\cos^2\theta\,\e^{-\int_0^t \d s\,\gamma_s}
\end{pmatrix}.
\end{equation}

The following quantities can be derived through direct calculation.
\begin{gather*}
\mathfrak{R}(\rho_0)=-\ln\frac{1+\sin2\theta}{2},\\
\mathfrak{R}(\rho_t)=-\ln\frac{1+\sin(2\theta)\e^{-\int_0^t\d s\,\frac{1}{2}\gamma_s}}{2},\\
\frac{\bra{f}\abs{\mathcal{L}(\rho_t)}\ket{f}}{\bra{f}\rho_t\ket{f}}=\frac{2\gamma_t\cos^2\theta\,\e^{-\int_0^t\d s\,\gamma_s}+\frac{1}{2}\gamma_t\abs{\sin(2\theta)}\e^{-\int_0^t\d s\,\frac{1}{2}\gamma_s}}{1+\sin(2\theta)\e^{-\int_0^t\d s\,\frac{1}{2}\gamma_s}}.
\end{gather*}
For the case of $\gamma_t=2$, the bound \eqref{bound QSL rugosity} becomes
\begin{equation}\label{bound_ru dissipative}
\begin{split}
\tau_{\text{ru}}^{\text{QSL}} &= \frac{\abs{\ln(1+\sin(2\theta)\e^{-\tau}) - \ln(1+\sin(2\theta))}}{\frac{1}{\tau}\int_0^\tau \d t \, \frac{4\cos^2\theta\e^{-2t}+\abs{\sin(2\theta)}\e^{-t}}{1+\sin(2\theta)\e^{-t}}}\\
&=\begin{cases}
\dfrac{\ln\frac{1+\sin(2\theta)}{1+\sin(2\theta)\e^{-2\tau}}}{2\cot\theta(1-\e^{-\tau})-\cot^2\theta\ln\frac{1+\sin(2\theta)}{1+\sin(2\theta)\e^{-2\tau}}}\tau,~& \sin(2\theta)>0,\\
\dfrac{\ln\frac{1+\sin(2\theta)\e^{-2\tau}}{1+\sin(2\theta)}}{2\cot\theta(1-\e^{-\tau})+(1+\csc^2\theta)\theta\ln\frac{1+\sin(2\theta)\e^{-2\tau}}{1+\sin(2\theta)}}\tau,~& \sin(2\theta)<0.
\end{cases}
\end{split}
\end{equation}

Owing to the symmetry of the sine function, it suffices to confine our analysis to the two intervals $(0,\frac{\pi}{2})$ and $(\frac{\pi}{2},\pi)$ of $\theta$. The quantum speed limits and the normalized relative errors for initial states in $(0,\pi/2)$ and $(\pi/2,\pi)$ of $\theta$ are displayed in FIG.~\ref{dissipative_ru}.
FIG.~\ref{dissipative_ru}(a) (resp. FIG.~\ref{dissipative_ru}(c)) can be roughly mapped to FIG.~\ref{dissipative_ru}(b) (resp. FIG.~\ref{dissipative_ru}(d)) under the transformation $\theta\rightarrow\pi-\theta$ and vice versa. Specifically, Figures~\ref{dissipative_ru}(a) and~\ref{dissipative_ru}(b) show that $\tau_{\text{ru}}^{\text{QSL}}\approx 0$ for all $\tau>0$ whenever  $0<\theta\lesssim \pi/5$ or $4\pi/5\lesssim\theta<\pi$, with $\tilde{\delta}_{\text{ru}}(\tau)\approx 1$ as illustrated in Figures~\ref{dissipative_ru}(c) and~\ref{dissipative_ru}(d). Consequently, as $\tau$ becomes large, the lower bound on the actual evolution time becomes increasingly loose, reducing to the trivial condition $\tau \gtrsim 0$. For a fixed $\theta$ satisfying $\pi/5\lesssim\theta<\pi/2$ or $\pi/2<\theta\lesssim 4\pi/5$, $\tau_{\text{ru}}^{\text{QSL}}$ varies monotonically as a function of the evolution time  $\tau>0$. In addition, Figures~\ref{dissipative_ru}(c) and~\ref{dissipative_ru}(d) imply that, for each initial-state parameter $\theta$, the normalized relative error $\tilde{\delta}_{\text{ru}}(\tau)$ remains approximately constant throughout the dynamical evolution. This indicates that the tightness of the bound in Eq.~\eqref{bound_ru dissipative} is preserved.

Here we observe that, as the real evolution time $\tau$ varies, the dissipative system exhibits fundamentally tighter bounds when the initial state is taken to be the maximum texture state of Eq.~(\ref{MTS}), in comparison with other initial states. This can be traced to the fact that Parisio, in Ref. \cite{parisio2024}, introduced state rugosity as a figure of merit for QST and proved that the Fourier states in Eq. ~(\ref{Fourier states}) possess maximal state rugosity. For a system with dimension $d=2$, there exists only one Fourier state, which corresponds to the input state with $\theta=\pi/2$. Together with the definition of the normalized relative error in Eq.~\ref{tightness_ru}, this yields a tighter lower bound on the evolution time.

\begin{figure}
\centering 
\includegraphics[width=0.5\textwidth]{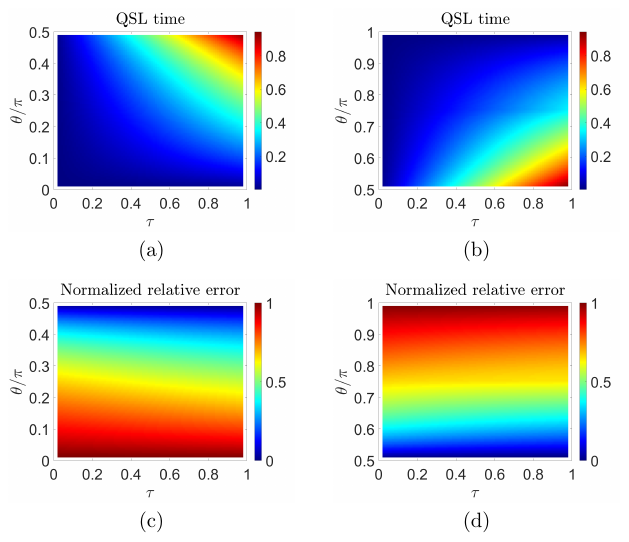}
\caption{The QSL time $\tau_{\text{ru}}^{\text{QSL}}$ and normalized relative error $\tilde{\delta}_{\text{ru}}(\tau)$ of dissipative dynamics corresponding to different initial states.}
\label{dissipative_ru}
\end{figure}

\subsection{QSL derived from texture Measure via Jensen-Shannon divergence}

In quantum information theory, the Jensen-Shannon divergence serves as an important quantity. For instance, it can be used to construct measures of quantum imaginarity \cite{tianPLA130479} and to study imaginarity measures for quantum channels \cite{fuQIP25140}. Here, we employ the QST measure defined by R. Muthuganesan based on the Jensen-Shannon divergence \cite{muthuganesan2025} to further explore the role of QST in the minimal evolution time required for the dynamical evolution of physical systems. As an application, we calculate the lower bound on the QST time for a system evolving under the amplitude damping channel.

We consider a general time-dependent nonunitary quantum evolution governed by $\rho_t = \mathcal{E}_t(\rho_0)$ for $t \in [0, \tau]$. The absolute value of the time derivative of quantum Jensen-Shannon divergence goes as follows,
\begin{align}
\abs{\frac{\d}{\d t}JS(f,\rho_t)}&\leqslant\abs{\frac{\d}{\d t}S(\omega_t)}+\frac{1}{2}\abs{\frac{\d}{\d t}S(\rho_t)} \notag \\
&=\frac{1}{2}\abs{\Tr\left(\frac{\d\rho_t}{\d t}\ln\omega_t\right)}+\frac{1}{2}\abs{\Tr\left(\frac{\d\rho_t}{\d t}\ln\rho_t\right)} \notag\\
&\leqslant\frac{1}{2}(\norm{\ln\omega_t}_\infty+\norm{\ln\rho_t}_\infty)\norm{\frac{\d\rho_t}{\d t}}_1 \notag \\
&=\frac{1}{2}\abs{\ln\left[\lambda_{\min}(\rho_t)\lambda_{\min}\left(\frac{f+\rho_t}{2}\right)\right]}\norm{\frac{\d\rho_t}{\d t}}_1, \label{derivativeJSD1}
\end{align}
where we have used the norm inequality $|\Tr(AB)|\leqslant\norm{A}_1\norm{B}_\infty$ and the equality  $\frac{\d}{\d t}S(\rho_t)=-\Tr(\frac{\d\rho_t}{\d t}\ln\rho_t)$.

Integrating Eq. \eqref{derivativeJSD1} over the time interval $t\in[0,\tau]$ yields an upper bound on the Jensen-Shannon divergence between the initial state $\rho_{0}$ and final state $\rho_{\tau}$
\begin{equation}\label{derivativeJSD2}
|JS(f,\rho_\tau)-JS(f,\rho_0)|\leqslant\int_0^\tau\d t\,h_{JS}(f,\rho_t)\norm{\frac{\d\rho_t}{\d t}}_1,
\end{equation}
where $h_{JS}(f,\rho_t):=\frac{1}{2}\abs{\ln[\lambda_{\min}(\rho_t)\lambda_{\min}(\omega_t)]}$ is an auxiliary function with $\omega_t=\frac{f+\rho_t}{2}$.

Given the explicit evolution form of a physical system $\rho_t = \mathcal{E}_t(\rho_0)$, the term $\left\|\frac{\d \rho_t}{\d t}\right\|_1$ on the right-hand side of Eq. \eqref{derivativeJSD2} can be directly computed, which quantifies the evolution speed of the system. Moreover, the weight function $h_{JS}(f,\rho_t)$ depends exclusively on $\lambda_{\min}(\rho_t)$ and $\lambda_{\min}(\omega_t)$, rendering it computationally tractable in practical implementations. It is worth noting that the left‑hand side of Eq. \eqref{derivativeJSD2} characterizes the distinguishability between the initial state $\rho_0$ and the final state $\rho_{\tau}$ through the absolute difference of their JSD‑based texture measures. This quantity is further upper‑bounded by the time integral of the product between $\left\|\frac{\d\rho_t}{\d t}\right\|_1$ and the weight function $h_{JS}(f,\rho_t)$. From this inequality, a quantitative relationship can be identified between the texture of instantaneous states and the evolution‑characterizing quantity $\left\|\frac{\d\rho_t}{\d t}\right\|_1$ during the system dynamics. This partially establishes a fundamental connection between the two quantities, implying that the texture resource inherent in quantum states imposes constraints on the minimal time required for physical system evolution.

To further evaluate the tightness of the bound on the JSD‑based measure derived in Eq.~(\ref{derivativeJSD2}), we define the following normalized relative error to quantify the tightness of the bound on JSD-based measure,
\begin{equation}\label{tightness}
\tilde{\delta}_{JS}(\tau) := \frac{\delta_{JS}(\tau)-\min(\delta_{JS})}{\max(\delta_{JS})-\min(\delta_{JS})},
\end{equation}
where $\delta_{JS}(\tau)$ denotes the unnormalized relative error,
\begin{equation*}
\delta_{JS}(\tau) := 1-\frac{\left|JS(f,\rho_\tau)-JS(f,\rho_0)\right|}{\Gamma},
\end{equation*}
with
\begin{equation*}
\Gamma = \int_0^\tau\d t\,h_{JS}(f,\rho_t)\norm{\frac{\d\rho_t}{\d t}}_1.
\end{equation*}
Here, $\Gamma$ is given by the function $h_{JS}(f,\rho_t)$, which depends on the smallest eigenvalue  
$\lambda_{\text{min}}(\rho_t)$ of the instantaneous state $\rho_t$, and the smallest eigenvalue  $\lambda_{\min}(\omega_t)$ of the convex combination $\omega_t$ of the textureless state $f$ and the instantaneous state $\rho_t$, as well as the quantity $\rho_t$.

It is straightforward to see that the normalized relative error defined in Eq.~(\ref{tightness}) mathematically captures the tightness of the inequality given in Eq.~(\ref{derivativeJSD2}). As will be shown later, this quantity also quantifies the discrepancy between the time lower bound for system evolution derived from the JSD‑based texture measure and the actual evolution time of the physical system. Subsequently, starting from Eq.~(\ref{derivativeJSD2}), we derive a family of QSLs linked to the JSD‑based texture for general non‑unitary dynamics. The following conclusion follows directly upon integrating both sides of Eq.~(\ref{derivativeJSD2}) over the interval $[0, \tau]$ and subsequently dividing by  $\tau$.

\begin{thm}\label{JSD QSL}
The evolution duration $\tau$ admits the following lower bound under a general time-dependent non-unitary evolution from an initial state $\rho_0$ to a final state $\rho_\tau$ within the time interval $t \in [0, \tau]$,
\begin{equation}\label{JSD bound QSL}
\tau \geqslant \tau_{JS}^{\mathrm{QSL}}\equiv\frac{\left|JS(f,\rho_\tau)-JS(f,\rho_0)\right|}{\left\langle\!\left\langle h_{JS}(f,\rho_t)\norm{\d\rho_t/\d t}_1 \right\rangle\!\right\rangle_\tau}
\end{equation}
with $\langle\!\langle \cdot \rangle\!\rangle_\tau=\frac{1}{\tau}\int_0^\tau\d t\,\cdot$.
\end{thm}

Therefore, the unnormalized relative error $\delta_{JS}(\tau)$ in Eq.~(\ref{tightness}) reads
$
\delta_{JS}(\tau) := 1-\frac{\tau_{JS}^{\text{QSL}}}{\tau}
$, which characterizes the tightness of the time lower bound given in Eq.~(\ref{JSD bound QSL}).

We start by considering the dynamical evolution governed by completely positive and trace-preserving (CPTP) quantum operations described by a time-dependent Kraus representation,
$\mathcal{E}_t(\rho)=\sum_\ell V_\ell(t) \rho V_\ell^\dagger(t)$ with $\sum_\ell V_\ell^\dagger(t) V_\ell(t)=\mathbb{I}$.
The state at time $t$ is $\rho_t=\mathcal{E}_t(\rho_0)$.
According to Ref. \cite{xuane00383}, the Schatten speed is upper bounded by $\| \d \rho_t / \d t \|_1 \leqslant 2\sum_\ell \| V_\ell \, \rho_0 (\d V_\ell^\dagger / \d t) \|_1$. Substituting this bound into the bound in Eq. (\ref{derivativeJSD2}), we obtain
\begin{equation}\label{channel1}
|JS(f,\rho_\tau)-JS(f,\rho_0)|\leqslant 2\sum_\ell\int_0^\tau\d t\,h_{JS}(f,\rho_t)\left\|V_\ell\,\rho_0\left(\frac{\d V_\ell^\dagger}{\d t}\right)\right\|_1.
\end{equation}
Based on Eq. (\ref{channel1}), the relative error can be expressed as
\begin{equation}\label{channel2}
\delta_{JS}(\tau)=1-\frac{|JS(f,\rho_\tau)-JS(f,\rho_0)|}{2\sum_\ell\int_0^\tau\d t\,h_{JS}(f,\rho_t)\|V_\ell\,\rho_0(\d V_\ell^\dagger/\d t)\|_1}.
\end{equation}

Integrating both sides of Eq. (\ref{channel1}) over the interval $[0,\tau]$ and then dividing by $\tau$, we get the QSL time formulated in Eq. (\ref{channel3}) for CPTP dynamics, as below,
\begin{equation}\label{channel3}
\tau\geqslant\tau_{JS}^{\mathrm{QSL}}=\frac{\left|JS(f,\rho_\tau)-JS(f,\rho_0)\right|}{2\sum_\ell\left\langle\!\!\left\langle h_{JS}(f,\rho_t)\|V_\ell\,\rho_0(\d V_\ell^\dagger/\d t)\|_1\right\rangle\!\!\right\rangle}.
\end{equation}

Eq.~\eqref{channel3} establishes a fundamental lower bound on the evolution time under the framework of CPTP dynamical maps. This bound is expressed in terms of the absolute change in the JSD-based texture measure, and explicitly involves the time-dependent Kraus operators $\{V_\ell\}$ that characterize the dynamics, the minimal eigenvalue $\lambda_{\min}(\rho_t)$ of the instantaneous state, as well as the minimal eigenvalue $\lambda_{\min}(\omega_t)$ of the convex combination between the textureless state $f$ and the instantaneous state $\rho_t$. In the subsequent analysis, we concentrate on evaluating the QSL bound given by Eq.~\eqref{channel3} and the relative error defined in Eq.\eqref{channel2} for single-qubit systems.

\textit{Amplitude damping channel}. We consider a quantum system undergoing  the amplitude damping channel with Kraus operators, $V_1 = |0\rangle\langle0| + \sqrt{1-\mu_t}\, |1\rangle\langle1|$, $V_2 = \sqrt{\mu_t}|0\rangle\langle1|$, $\mu_t=1-\e^{-\gamma t}$. The system is initially prepared in a general single-qubit state as $
\rho_0 = \frac{\mathbb{I} + \mathbf{r}_0 \cdot \boldsymbol{\sigma}}{2}$, where the Bloch vector $\mathbf{r}_0=\{r_0\sin\theta\cos\phi,r_0\sin\theta\sin\phi,r_0\cos\theta\}$ with $r_0\in [0,1]$, $\theta\in [0,\pi]$ and $\phi\in[0,2\pi]$, and $\boldsymbol{\sigma}=\{\sigma_x,\sigma_y,\sigma_z\}$ is given by the standard Pauli matrices. The Bloch representation of the textureless state $f$ is $f=\frac{\mathbb{I}+\mathbf{r}_f\cdot\boldsymbol{\sigma}}{2}$, where $\mathbf{r}_f=(1,0,0)$.

The evolved state becomes
\begin{equation*}
\begin{aligned}
\rho_t^{AD}&=\frac{1}{2}
\begin{pmatrix}
1+[\mu_t+(1-\mu_t)r_0\cos\theta] & \sqrt{1-\mu_t}(r_0\sin\theta\cos\phi-\mathrm{i}r_0\sin\theta\sin\phi) \\
 \sqrt{1-\mu_t}(r_0\sin\theta\cos\phi+\mathrm{i}r_0\sin\theta\sin\phi) & 1-[\mu_t+(1-\mu_t)r_0\cos\theta]
\end{pmatrix} \\
&=\frac{\mathbb{I}+\mathbf{r}_t^{AD}\cdot\boldsymbol{\sigma}}{2},
\end{aligned}
\end{equation*}
with $\mathbf{r}_t^{AD}=(\sqrt{1-\mu_t}r_0\sin\theta\cos\phi,\sqrt{1-\mu_t}r_0\sin\theta\sin\phi,\mu_t+(1-\mu_t)r_0\cos\theta)$. The instantaneous state has eigenvalues $\lambda_{\max/\min}(\rho_t^{AD})=(1\pm r_t^{AD})/2$, with
$
r_t^{AD}=\sqrt{(1-\mu_t)r_0^2\sin^2\theta+[\mu_t+(1-\mu_t)r_0\cos\theta]^2}.
$

For the convex combination $\omega_t^{AD}=(f+\rho_t^{AD})/2=(\mathbb{I}+\mathbf{s}_t^{AD}\cdot\boldsymbol{\sigma})/2$, with
\begin{equation*}
\mathbf{s}_t^{AD}=\frac{\mathbf{r}_f+\mathbf{r}_t^{AD}}{2}
=\frac{1}{2}\Big(1+\sqrt{1-\mu_t}r_0\sin\theta\cos\phi,\sqrt{1-\mu_t}r_0\sin\theta\sin\phi,\mu_t+(1-\mu_t)r_0\cos\theta\Big),
\end{equation*}
the eigenvalues are $\lambda_{\max/\min}(\omega_t)=(1\pm s_t^{AD})/2$, where
\begin{equation*}
s_t^{AD}=\frac{1}{2}\sqrt{1+2\sqrt{1-\mu_t}r_0\sin\theta\cos\phi+(1-\mu_t)r_0^2\sin^2\theta+[\mu_t+(1-\mu_t)r_0\cos\theta]^2}.
\end{equation*}
And when $t=0$, $\lambda_{\max/\min}(\omega_0)=(1\pm s_0)/2$, with
$
s_0=\frac{1}{2}\sqrt{1+2r_0\sin\theta\cos\phi+r_0^2}.
$

Therefore, we obtain the JSD-based measures for initial state $\rho_0$ and final state $\rho_\tau$, the auxiliary function $h_{JS}(f,\rho_t)$, and the Schatten speed $\sum_\ell\|V_\ell\,\rho_0(\d V_\ell^\dagger/\d t)\|_1$ as follows,
\begin{gather}
JS(f,\rho_0)=\frac{1-r_0}{4}\ln\frac{1-r_0}{2}+\frac{1+r_0}{4}\ln\frac{1+r_0}{2}-\frac{1-s_0}{2}\ln\frac{1-s_0}{2}-\frac{1+s_0}{2}\ln\frac{1+s_0}{2}, \label{ADquantity1} \\
JS(f,\rho_\tau^{AD})=\frac{1-r_\tau^{AD}}{4}\ln\frac{1-r_\tau^{AD}}{2}+\frac{1+r_\tau^{AD}}{4}\ln\frac{1+r_\tau^{AD}}{2}-\frac{1-s_\tau^{AD}}{2}\ln\frac{1-s_\tau^{AD}}{2}-\frac{1+s_\tau^{AD}}{2}\ln\frac{1+s_\tau^{AD}}{2}, \label{ADquantity2} \\
h_{JS}(f,\rho_t^{AD})=\frac{1}{2}\abs{\ln(\frac{1-r_t^{AD}}{2}\cdot\frac{1-s_t^{AD}}{2})}, \label{ADquantity3} \\
\sum_\ell\|V_\ell\,\rho_0(\d V_\ell^\dagger/\d t)\|_1=\frac{1}{4}\gamma\e^{-\gamma t}\Big(1-r_0\cos\theta+\sqrt{(1-r_0\cos\theta)^2+\e^{\gamma t}r_0^2\sin^2\theta}\,\Big) \label{ADquantity4}.
\end{gather}

Substituting Eqs.~(\ref{ADquantity1}-\ref{ADquantity4}) into Eq.~\eqref{channel3} and Eq.~(\ref{tightness}), we derive the lower bound on QSL time for amplitude damping channel. We present the numerical results for the QSL time and the normalized relative error corresponding to different initial states with decay rate $\gamma=0.8$ in FIG.~\ref{figureADchannel} and FIG.~\ref{errorfigureADchannel}.
\begin{figure}
\centering
\includegraphics{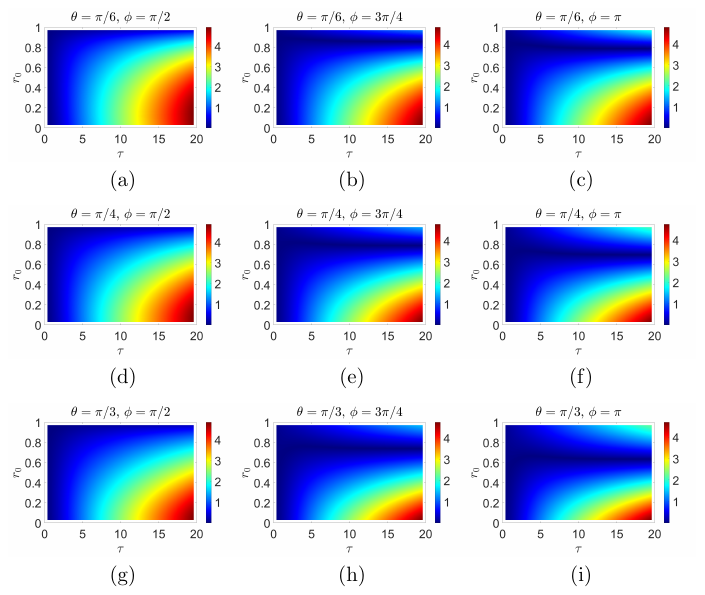}
\caption{Density plots of the texture measure-based QSL time $\tau_{JS}^{\text{QSL}}$ as a function of the parameters $\tau$ and $r\in[0,1)$,  with $\theta\in\{\pi/6,\pi/4,\pi/3\}$ and $\phi\in\{\pi/2,3\pi/4,\pi\}$ for initial single-qubit states subject to the amplitude damping channel.}
\label{figureADchannel}
\end{figure}

\begin{figure}
\centering
\includegraphics{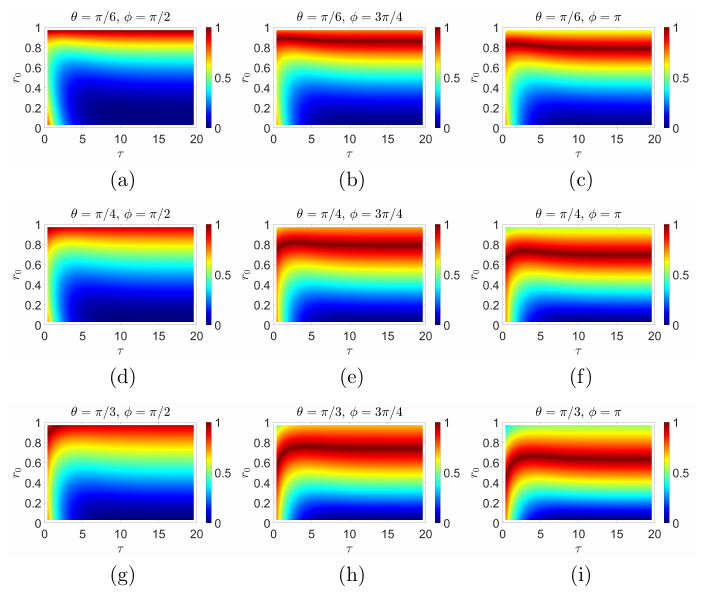}
\caption{Density plots of the normalized relative error $\tilde{\delta}_{JS}(\tau)$ as a function of the parameters $\tau$ and $r\in[0,1)$, with $\theta\in\{\pi/6,\pi/4,\pi/3\}$ and $\phi\in\{\pi/2,3\pi/4,\pi\}$ for initial single-qubit states subject to the amplitude damping channel.}
\label{errorfigureADchannel}
\end{figure}

We note in FIG.~\ref{figureADchannel} that, for a fixed $\phi\in\{\pi/2,3\pi/4,\pi\}$, $\tau_{JS}^{\text{QSL}}$ exhibits the similar qualitative behavior for different $\theta\in\{\pi/6,\pi/4,\pi/3\}$, despite slight quantitative differences among these cases, while for a fixed value $r_0 \in [0,1)$ of the mixedness parameter, $\tau_{JS}^{\text{QSL}}$ varies monotonically as a function of the parameter $\tau > 0$. In particular, for $\phi=3\pi/4$ and $\theta\in\{\pi/6,\pi/4,\pi/3\}$, Figures~\ref{figureADchannel}(b),~\ref{figureADchannel}(e) and~\ref{figureADchannel}(h) show that $\tau_{JS}^{\text{QSL}}\approx 0$ for all $\tau>0$ when $0.6 \lesssim r_0 < 1$. The analogous phenomenon also holds for $\phi=\pi$ and $\theta\in\{\pi/6,\pi/4,\pi/3\}$, which is evidenced by the results presented in FIG.~\ref{figureADchannel}(c),~\ref{figureADchannel}(f) and~\ref{figureADchannel}(i).

In FIG.~\ref{errorfigureADchannel}, we note that, for a fixed mixedness parameter $0 < r_0 \lesssim 0.6$ in FIG.~\ref{errorfigureADchannel}(a), $0 < r_0 \lesssim 0.5$ in FIG.~\ref{errorfigureADchannel}(d), $0 < r_0 \lesssim 0.4$ in FIG.~\ref{errorfigureADchannel}(g), $0 < r_0 \lesssim 0.4$ in FIG.~\ref{errorfigureADchannel}(b), $0 < r_0 \lesssim 0.3$ in FIG.~\ref{errorfigureADchannel}(e), $0 < r_0 \lesssim 0.2$ in FIG.~\ref{errorfigureADchannel}(h), $0 < r_0 \lesssim 0.3$ in FIG.~\ref{errorfigureADchannel}(c), $0 < r_0 \lesssim 0.2$ in FIG.~\ref{errorfigureADchannel}(f), $0 < r_0 \lesssim 0.2$ in FIG.~\ref{errorfigureADchannel}(i), respectively, $0 \lesssim \tilde{\delta}_{JS}({\tau})\lesssim 0.2$ for $\tau\gtrsim 5$, which indicates that as $\tau$ becomes larger, the lower bound on the actual evolution time is tighter. And for a fixed mixedness parameter $0.8 \lesssim r_0 \lesssim 0.9$ in FIG.~\ref{errorfigureADchannel}(b), $0.7 \lesssim r_0 \lesssim 0.8$ in FIG.~\ref{errorfigureADchannel}(e), $0.6 \lesssim r_0 \lesssim 0.7$ in FIG.~\ref{errorfigureADchannel}(h), $0.7 \lesssim r_0 \lesssim 0.9$ in FIG.~\ref{errorfigureADchannel}(c), $0.6 \lesssim r_0 \lesssim 0.8$ in FIG.~\ref{errorfigureADchannel}(f), $0.5 \lesssim r_0 \lesssim 0.7$ in FIG.~\ref{errorfigureADchannel}(c), respectively, $0.8 \lesssim \tilde{\delta}_{JS}({\tau})\lesssim 1$ for  all $\tau > 0$, which indicates that as $\tau$ becomes larger, the lower bound on the actual evolution time becomes increasingly loose, effectively yielding the trivial condition $\tau \gtrsim 0$. Furthermore, from a cross-comparison perspective across the multiple figures, we conclude that for a fixed $\theta$, the larger the value of $\phi$, the larger the normalized relative error.

\section{Conclusion}\label{sect4}
We have introduced and investigated quantum speed limits derived from quantum‑state texture (QST) variations, establishing general lower bounds on the evolution time of quantum systems. For the QST measure defined via trace distance, we illustrate the relation between QST and QSLs using dephasing dynamics, while dissipative dynamics is adopted as an exemplary platform for state rugosity. Furthermore, for the Jensen‑Shannon divergence‑induced QST, we derive time lower bounds for systems evolving under completely positive and trace‑preserving dynamics and perform quantitative calculations for the amplitude‑damping channel. Notably, our derived bounds are readily applicable to other canonical open‑system models, including the phase damping and depolarizing channels. Although we illustrate our results with examples for two‑dimensional systems, the QSL lower bounds established in this paper are also applicable to high‑dimensional systems, such as the erasure channel, the $d$-dimensional dephasing channel, the $d$-dimensional depolarizing channel, and the Werner--Holevo channel \cite{WildeBook}, as well as the  local amplitude damping channel with finite temperature \cite{GuoQIP1399,GuoQIP2851}. It would be of interest to generalize our results in composite systems. The QSL relations obtained in this work may find promising applications in quantum thermodynamics, quantum control theory, and practical quantum engineering technologies.

\section*{Acknowledgments:}

S. M. Fei acknowledges the financial support from specific research fund of the Innovation Platform for Academicians of Hainan Province.

\end{document}